\documentclass[runningheads,envcountsame]{llncs}
\usepackage[utf8]{inputenc}
\usepackage{amsmath}
\usepackage{amsfonts}
\usepackage{amssymb}
\usepackage{proof}
\usepackage{graphicx}
\usepackage{color}
\usepackage[dvipsnames]{xcolor}
\usepackage{colortbl}
\usepackage{xspace}
\usepackage{mycommands}
\usepackage[show]{ed}
\usepackage{marginnote}

\title{A Game Model for Proofs with Costs}
\titlerunning{Derivations with Costs}
\author{Timo Lang\inst{1} and Carlos Olarte\inst{2} and Elaine Pimentel\inst{2}\\
 and Christian G. Ferm\"{u}ller\inst{1}
\thanks{Olarte and Pimentel are funded by CNPq, CAPES and the project FWF START Y544-N23. Lang is supported by FWF project W1255-N23.}}
\institute{
TU-Wien, Austria \and
Universidade Federal do Rio Grande do Norte,  Brazil}
\authorrunning{Lang, Olarte, Pimentel and Ferm\"{u}ller}

\begin{document}

\maketitle

\begin{abstract}

We look at substructural calculi from a game semantic point of view,
guided by certain intuitions about resource conscious and, more
specifically, cost conscious reasoning. To this aim, we start with a
game, where player \I
defends a claim corresponding to a (single-conclusion) sequent, while
player \II tries to refute that claim. Branching rules for additive
connectives are modeled by choices of \II, while branching for
multiplicative connectives leads to splitting the game into parallel
subgames, all of which have to be won by player \I to succeed. The game
comes into full swing by adding cost labels to assumptions, and a corresponding budget.
 Different proofs of the same end-sequent are interpreted as
more or less expensive strategies for \I to defend the corresponding claim.
This leads to a new kind
of labelled calculus, which can be seen as a fragment of \SELL\ (subexponential linear logic).  
Finally, we generalize the
concept of costs in proofs by using a semiring
structure, illustrate our interpretation by examples and investigate
some proof-theoretical properties.

 \end{abstract}

\section{Introduction}

Various kinds of game semantics have been introduced to characterize computational features of
substructural logics, in particular fragments and variants of linear logic (\LL) \cite{girard87tcs}. This line of research
can be traced back to the works of   Blass~\cite{DBLP:journals/apal/Blass92,DBLP:journals/igpl/Blass97}, Abramsky and Jagadeesan \cite{DBLP:journals/jsyml/AbramskyJ94},
Hyland and Ong \cite{Hyland93fairgames}, Lamarche \cite{DBLP:conf/lics/Lamarche95}, Japaridze \cite{DBLP:journals/apal/Japaridze97}, Melli\`{e}s \cite{DBLP:conf/lics/Mellies05}, Delande et al.~\cite{DBLP:journals/apal/DelandeMS10},
among several others. 

Our particular view of game semantics is that it is not just a technical tool for characterizing provability in certain calculi,
but rather a playground for illuminating specific semantic intuitions underlying certain
proof systems. 
Specially, we aim at a better understanding of \emph{resource conscious} reasoning, 
which is often cited as a motivation for substructural logics.

In a first step, we characterize a version of linear logic (exponential-free affine inuitionistic linear logic~$\mathbf{aIMALL}$, or, equivalently, Full Lambek Calculus
with exchange and weakening $\FLew$) by a game, where the difference between additive and multiplicative
connectives is modeled as sequential versus parallel continuation in game states that directly 
correspond to sequents.
 More precisely, every branching rule for a multiplicative connective
  corresponds to a game rule that splits the
current run of the game into two independent subgames. Player \I, who seeks to establish the
validity of a given sequent, has to win all the resulting subgames. In contrast,
a branching rule for an additive connective is modeled by a choice of player \II between
two possible succeeding game states,  corresponding to the  premises of the sequent rule in question. 
Note that this amounts to a deviation from the paradigm ``formulas as games'', underlying
the game semantic tradition initiated by Blass~\cite{DBLP:journals/apal/Blass92}. Our games are, at least structurally,
closer to Lorenzen's game for intuitionistic logic~\cite{Lorenzen1960-LORLUA}, where a state roughly corresponds to a situation
in which a proponent seeks to defend a particular statement against attacks from an opponent, who,
in general, has already granted a bunch of other statements.  This kind of semantics
for linear logic (but without the sequential/parallel distinction) was first explored in  \cite{DBLP:conf/tableaux/FermullerL17}. 

As long as we only care about the existence of winning strategies, the distinction between sequential and
parallel subgames 
is
redundant.
However, our model not only highlights the
intended semantics, but it also has concrete effects once we introduce \emph{prices} for
resources (represented by formulas) into the game. This is done via 
unary operators $\pnbang{a}$ and $\vnbang{a}$, $a\in\real^+$, which share some characteristic features with \emph{subexponentials} in \LL\  (\SELL~\cite{danos93kgc,nigam09ppdp}). The intuition is that a formula $\vnbang{a}A$ is a \emph{single use resource with price $a$}: By paying $a$, we can ``unpack''~$\vnbang{a}A$ to obtain the formula~$A$, and~$\vnbang{a}A$ is destroyed in the process. On the other hand,~$\pnbang{a}A$ denotes a \emph{permanent resource}: From~$\pnbang{a}A$ we can obtain~$A$ as often as we want, each time paying the price~$a$.
We lift our game to the extended language by enriching game states with a \emph{budget} that is decreased whenever
a price is paid. Different strategies for proving the same endsequent can then be compared by the budget which they require to be run safely, i.e. without getting into debts.
This form of resource consciousness not only enhances the game, but it also translates
into a novel sequent system, where cost bounds for proofs are attached as labels to sequents. 

We observe that, up to this point, we only considered resources in {\em assumptions}. This is translated to sequents by restricting  {\em negatively} the occurrences of the modalities $\pnbang{a}$ and $\vnbang{a}$. Thus a promotion rule is not present and the proof-theoretic properties of the proposed systems, such as cut-elimination, can be mimicked by the ones of $\mathbf{aIMALL}$. 
We hence move towards two possible generalizations. First, we propose a  broader notion of cost and prices (for both the game and corresponding calculi) beyond the domain of the non-negative real numbers. For this, we organize the labels/prices in a semiring structure that enables for the instantiation of several interesting concrete examples, having the same game-theoretic characterization. Second, we discuss the quest of allowing modalities also in positive contexts, showing the limitations of such approach.

\noindent{\em Organization and contributions. } Section~\ref{sec:pre} defines the basic game for $\mathbf{aIMALL}$ and establishes the correspondence between winning strategies and proofs. 
Section \ref{sec:costs} introduces the concept of prices and budgets into the game. 
The existence of cost-minimal strategies is shown in Section \ref{sec:spectrum} and cut-admissibility is discussed in Section \ref{sec:cut}. In Section~\ref{sec:general}, the concept of prices is generalized and several examples of our interpretation of costs in proofs are presented. 
In Section \ref{section:pos}, the challenge of extending the semantics to full subexponential linear logic is discussed. Section~\ref{sec:conc} concludes the paper.  
\section{A game model of branching}
\label{sec:pre}
Our starting point is a calculus for \emph{affine intuitionistic linear logic without exponentials} ($\mathbf{aIMALL}$)~\cite{girard87tcs}, whose calculus is also equivalent to \FLew, the \emph{Full Lambek calculus with exchange and weakening}. We denote this calculus simply by \aIMALL for brevity. Formulas in~\aIMALL are built from the grammar 

\[
 A ::= p \mid \zero \mid \one  
 \mid A_1 \lolli A_2
\mid   A_1 \tensor A_2 \mid A_1 \with A_2  \mid A_1 
 \oplus A_2 .
\]

\noindent
where $p$ stands for atomic propositions (variables); $\zero/\one$ are the false/true units;  $\limp$ denotes linear implication;  $\otimes/\with$ are the multiplicative/additive conjunctions; and $\oplus$ is the additive disjunction. 

We shall use $A,B,C$ (resp. $\Gamma,\Delta$) to range over formulas (resp. multisets of formulas).  The rules 
are in Fig.~\ref{fig:ll}. Note that the cut rule is not included in our presentation of $\aIMALL$ and that weakening is present only implicitly, via the context $\Gamma$ in the initial sequents. Furthermore, in rule $I$, $p$ is a propositional variable.
We shall write $\vdash_\aIMALL S$ if the sequent $S$ is provable in $\aIMALL$.
\begin{figure}[t]
\resizebox{\textwidth}{!}{
$
\begin{array}{c}
\hline\mbox{Sequent System for }\aIMALL\\\hline\\
 \infer[\tensor_L]{\Gamma, A \tensor B \lra C}
{\Gamma, A, B \lra C} 
\quad 
\infer[\tensor_R]{\Delta_1, \Delta_2 \lra A \tensor B}
{\Delta_1 \lra A & \Delta_2 \lra B}
\quad
\infer[\lolli_L]{\Delta_1, \Delta_2, A \lolli B \lra C}
{\Delta_1 \lra A & \Delta_2, B \lra C}
\quad 
\infer[\lolli_R]{\Gamma \lra A \lolli B}{\Gamma, A \lra B}
\\\\
 \infer[\with_{L_i}]{\Gamma, A_1 \with A_2 \lra B}
{\Gamma, A_i\lra B} 
\quad 
\infer[\with_R]{\Gamma \lra A \with B}
{\Gamma \lra A & \Gamma \lra B}
\quad
\infer[\oplus_L]{\Gamma, A \oplus B \lra C}
{\Gamma, A \lra C & \Gamma, B \lra C}
\quad 
\infer[\oplus_{R_i}]{\Gamma \lra A_1 \oplus A_2}{\Gamma \lra A_i}
\\\\
\infer[I]{\Gamma,p \lra p}{} 
 \qquad
\infer[\one_R]{\Gamma \lra \one}{}
\qquad 
\infer[\zero_L]{ \Gamma,\zero\lra A}{}
\\\\
\hline\mbox{Sequent System for }\aIMALLR\\\hline\\
\infer[\tensor_R]{\pnbang{}\Gamma,\Delta_1, \Delta_2 \lra A \tensor B}
{\pnbang{}\Gamma,\Delta_1 \lra A & \pnbang{}\Gamma,\Delta_2 \lra B}
\quad
\infer[\lolli_L]{\pnbang{}\Gamma,\Delta_1, \Delta_2, A \lolli B \lra C}
{\pnbang{}\Gamma,\Delta_1 \lra A & \pnbang{}\Gamma,\Delta_2, B \lra C}
\\\\
\infer[\pnbang{}_L]{\Gamma,\pnbang{a}A\lra C}{\Gamma,\pnbang{a}A,A\lra C}
\qquad
\infer[\vnbang{}_L]{\Gamma,\vnbang{a}A\lra C}{\Gamma,A\lra C}
\end{array}
$}\caption{Sequent systems $\aIMALL$ and $\aIMALLR$}
\label{fig:ll}
\end{figure}

We shall  characterize \aIMALL proofs as winning strategies (w.s.) in a certain game. Usually, one can interpret bottom-up proof search in sequent systems as a  game, where  at any given state, player \I first 
chooses a formula of a sequent and, in the next step,  either \I
moves to the premise sequent of the corresponding introduction rule (if the rule has
only one premise); or player \II
chooses 
a premise sequent in which the game continues (if the rule has more than one premise).
Alternatively, 
rather than letting player \II choose the subgame,
one may stipulate that the game splits into independent subgames, all of
which player \I has to win. 
At first glance, these two approaches 
might seem  different. However, the difference is only of interpretation and it does not affect 
the (non-)existence of w.s.'s for~\I. To see this, note that, by definition of a w.s.,
player \I has to be prepared to answer to every possible choice of her opponent~\II. 
Therefore, it
does not matter whether we require \I to actually win every subgame 
or whether we image \I to play a single run
where she wins 
irrespectively of \II's  choices. 
Hence, the two  interpretations are equivalent in terms of \I's w.s.'s but 
they provide different viewpoints of branching sequent rules. Going more into detail, we can see that this equivalence holds as long as the parallel games are \emph{independent}. We will break this independence later on by introducing a budget which is shared among parallel games (see Section \ref{sec:costs}).

The distinguishing feature of the game 
\GAIMALL  below is: 
branching
in {\em additive rules} is modeled as choices of \II, whereas in branching {\em multiplicative rules},  \I   splits the context into two disjoint parts,
which then form the corresponding contexts of two  subgames to
be played in parallel. Consequently, a state of the game is represented
by a \textit{multiset of sequents}, each belonging to a separate subgame. 

\begin{definition}[The game \GAIMALL]\label{definition:GAILL}
\GAIMALL is a game of two players, \I and \II. Game states (denoted by  $G,H$) are finite multisets of sequents.
\GAIMALL proceeds in rounds, initiated by \I's selection of a sequent $S$ from the current game state. The successor state is determined according to rules that fit one of the  following schemes:

\[
\begin{array}{lllll}
{(1)} &G\cup\{S\}&\quad\leadsto\quad&  \quad G\cup\{S'\} & \qquad \\
{(2)} &G\cup\{S\}&\quad\leadsto\quad&  \quad G\cup\{S_1\}\cup\{S_2\} & \qquad 
\end{array}
\]

\noindent In (1), the subgame $S$ changes to $S'$.
In (2), the subgame $S$ splits into two subgames $S_1$ and $S_2$. Here is the complete description of a round: After \I has chosen a sequent $S$ among the current game state, she chooses a principal formula in $S$ and a matching rule instance $r$ of \aIMALL such that $S$ is the conclusion of that rule. Depending on $r$, the round proceeds as follows:
\begin{enumerate}
\item If $r$ is a unary rule with premise $S'$, then the game proceeds in the game state $G\cup\{S'\}$ (no interaction of \II is required).
\item {\bf Parallelism:} If $r$ is a binary rule with premises 
$S_1,S_2$
pertaining to a \emph{multiplicative} connective, then the game proceeds in the game state $G\cup\{S_1\}\cup\{S_2\}$ (again, no interaction of \II is required).
\item {\bf \II-choice:} If $r$ is a binary rule with premises $S_1,S_2$ pertaining to an \emph{additive} connective, then \II chooses $S'\in\{S_1,S_2\}$ and the game proceeds in the game state $G\cup\{S'\}$.
\end{enumerate}

\noindent A {\bf winning state} (for \I) is a game state consisting of initial sequents of $\aIMALL$ only, 
that is, sequents having one of the forms $(\Gamma,p\lra p)$, $(\Gamma,\zero\lra A)$, $(\Gamma\lra\one)$. 
\end{definition}

\begin{example}
As an example of a round in $\GAIMALL$, assume that the game starts with 
$
\Delta\lra A\tensor B
$. 
 \I might select $A\tensor B$ as the principal formula. For the choice of a matching instance of the rule $\tensor_R$, she also has to choose a partition $\Delta=\Delta_1\cup\Delta_2$. The game then continues in the state
 $
\{(\Delta_1\lra A),(\Delta_2\lra B)\}
$. 
\end{example}
The following definitions are standard in game theory.
\begin{definition}[Plays and strategies]
A {\bf play} of \GAIMALL on a game state $H$ is a sequence $H_1,H_2,\ldots,H_n$ of game states, where $H_1=H$ and each $H_{i+1}$ arises by playing one round on $H_i$.
 A {\bf strategy} (for \I) on a game state $H$ is defined as a function telling \I how to move in any given state. 
 A strategy on $H$ is a {\bf winning strategy (w.s.)} if all plays following it eventually reach a winning state.
We shall write $\wins{}{\GAIMALL} H$ if $\I$ has a w.s. on the game state $H$. 
\end{definition}

Given w.s.'s $\pi_1,\ldots,\pi_n$ for sequents $S_1,\ldots,S_n$, there is an obvious w.s. for the game state $\{S_1,\ldots,S_n\}$ which could be specified as ``play according to $\pi_i$ in the subgame $S_i$''. Not all w.s.'s for $\{S_1,\ldots,S_n\}$ need to arise in such a way though, since in principle it is allowed that moves in a subgame $S_i$ depend on the moves in another subgame $S_j$. Nevertheless, since in the game $\GAIMALL$ valid moves and the winning conditions in all subgames are independent, we can restrict to strategies of the former kind. This observation is encapsulated as follows. 
\begin{lemma}[Independence
]
\label{lemma:independency}
$\wins{}{\GAIMALL}\{S_1,\ldots,S_n\}\quad\iff \quad\text{for all $i\leq n$, }\wins{}{\GAIMALL}S_i$
\end{lemma}
Strategies in a game can be pictured as trees of game states, and therefore strategies share a common form with proofs. In our case, game states are multisets of sequents. However, by virtue of the above lemma, we obtain a notation of winning strategies which uses single sequents as nodes, at least if the initial state of the game is a sequent.
\begin{theorem}[Adequacy for $\GAIMALL$]
\label{theorem:adeq1}
Let $S$ be a sequent. Then 
$\wins{}{\GAIMALL}\{S\}~\iff~\vdash_\aIMALL S$.
\end{theorem}
\begin{proof}
($\Leftarrow$) is a straightforward induction on the length of proofs. ($\Rightarrow$) is proved by induction on a w.s. (the maximal number of moves which can occur following it). We only present the case where Lemma \ref{lemma:independency} comes into play.
Assume that the state is $\Delta_1,\Delta_2\lra A\tensor B$  and $\pi$ tells  \I to choose the instance of $\tensor_R$ with premises $\Delta_1\lra A$ and $\Delta_2\lra B$. 
 By {\bf parallelism}, the successor state is $\{(\Delta_1\lra A),(\Delta_2\lra B)\}$. Since $\pi$ is a w.s., it must contain a substrategy $\pi'$ for $\{(\Delta_1\lra A),(\Delta_2\lra B)\}$. By Lemma \ref{lemma:independency}, we may assume that $\pi'$ is of the form: 
``Use $\pi_1$ to play in the subgame $\Delta_1\lra A$ and $\pi_2$ to play in $\Delta_2\lra B$'' for some w.s.'s $\pi_1,\pi_2$ for $\Delta_1\lra A$ and $\Delta_2\lra B$ respectively. By induction, there are $\aIMALL$-proofs 
$\pr_1,\pr_2$  for the sequents $\Delta_1\lra A$ and $\Delta_2\lra B$.  
Applying $\tensor_R$ below 
$\pr_1$ and $\pr_2$, we obtain a $\aIMALL$-proof 
$\pr$ of $\Delta_1,\Delta_2\lra A\tensor B$.
\end{proof} 
\section{Adding costs}\label{sec:costs}
To increase the expressiveness of our framework, we now augment assumptions with costs, where assumptions are formulas occurring {\em negatively} on sequents. Costs will be modeled---for now---by elements of $\real^+$, the set of non-negative real numbers. Formally, we add the unary modal operators $\pnbang{a}$ and $\vnbang{a}$
for each $a\in\real^+$ to our language and call the resulting formulas \emph{extended formulas}. {An extended formula $\vnbang{a}A$ can be considered as a \emph{single use resource with price $a$}: By paying $a$, we can ``unpack''~$\vnbang{a}A$ to $A$ (and~$\vnbang{a}A$ is destroyed in the process). On the other hand,~$\pnbang{a}A$ is a \emph{permanent resource}: We can obtain as many copies of~$A$ from it as we want, each time paying the price~$a$.}
\begin{definition}\label{def:holecont} 
An \emph{extended sequent} is a sequent built from extended formulas in which subformulas $\pnbang{a}A$ and $\vnbang{a}A$ occur only in negative polarity.
\end{definition}
The notion of polarity  is the standard one: A subformula occurrence in the antecedent of a sequent is {\em negative} if it occurs in the scope of an even number (including $0$) of contexts $([\cdot]\limp B)$, and otherwise it is {\em positive}. For occurrences of a subformula in the consequent, one replaces ``even'' by ``odd''. For instance, $\vnbang{a}p \otimes p', (\pnbang{b}q\limp q')\limp q'' \lra \pnbang{c}r\limp r'$ is an  
extended sequent.
We denote by $\pnbang{}\Gamma$ a set of formulas  prefixed with $\pnbang{a}$ for some (not necessarily the same) $a\in\real^+$. We  introduce a game $\GAIMALLR$ 
similarly as we did for  $\GAIMALL$. The rules of $\GAIMALLR$ make reference to the calculus $\aIMALLR$ of Fig. \ref{fig:ll}. It is obtained by interpreting all sequents as extended sequents, replacing the rules $\tensor_R$ and $\limp_L$ as indicated in Fig. \ref{fig:ll} (for internalizing contraction) and adding the \emph{dereliction} rules 
\[
\infer[\pnbang{}_L]{\Gamma,\pnbang{a}A\lra C}{\Gamma,\pnbang{a}A,A\lra C}
\qquad
\infer[\vnbang{}_L]{\Gamma,\vnbang{a}A\lra C}{\Gamma,A\lra C}
\]
Note that there is no right rules for $\pnbang{}$ and $\vnbang{}$ in $\aIMALLR$ since they only appear in negative polarity.
\begin{remark}\label{remark:fragment}
$\aIMALLR$ can be naturally seen as a fragment of subexponential linear logic (\SELL~\cite{danos93kgc}). More specifically, let $\aSELLR$ be a single conclusion calculus for \SELL with weakening, and let $\Sigma=\langle \realbu,\preceq,\mathcal{U}\rangle$ be the subexponential signature  where the set of unbounded subexponentials (that can be weakened and contracted at will) is $\mathcal{U}=\{(a,\mathfrak{u}) \mid a\in \real^+\}$, and
$\preceq$ is any partial order on $\realbu$ in which, as standardly required in \SELL,  no bounded subexponential is above an unbounded one. 
We identify the subexponential $!^{(a,\mathfrak{b})}$ with $\vnbang{a}$ and   $!^{(a,\mathfrak{u})}$ with $\pnbang{a}$. Then $\aIMALLR$ is precisely the subsystem of $\aSELLR$ given by the syntactic restriction that subexponentials occur only in negative polarity. We will exploit this relation between $\aIMALLR$ and $\aSELLR$ later in Section \ref{sec:cut}. For some remarks on the system without the syntactic restriction, see Section~\ref{section:pos}.\end{remark}

Let us return to the game now. The main difference between $\GAIMALL$ and $\GAIMALLR$ is that game states in the latter will involve a {\bf budget} (modeled as a real number) which will decrease whenever rules $\pnbang{}_L$ and $\vnbang{}_L$ are invoked.

\begin{definition}[The game $\GAIMALLR$]\label{definition:GAIMALLR}
$\GAIMALLR$ is a game of two players, \I and \II. Game states are tuples $(H,b)$, where $H$ is a finite multiset of extended sequents and $b\in\real$ is a ``budget''.
\GAIMALL proceeds in rounds, initiated by \I's selection of an extended sequent $S$ from the current game state. The successor state is determined according to rules that fit one of the two following schemes:

\[
\begin{array}{lllll}
{(1)} &(G\cup\{S\},b)&\quad\leadsto\quad&  \quad (G\cup\{S'\},b') & \\
{(2)} &(G\cup\{S\},b)&\quad\leadsto\quad&  \quad (G\cup\{S^1\}\cup\{S^2\},b)
\end{array}
\]

A round proceeds as follows: After \I has chosen an extended sequent $S\in H$ among the current game state, she chooses a rule  instance  $r$ of $\aIMALLR$ such that $S$ is the conclusion of that rule. Depending on  $r$, the round proceeds as follows:
\begin{enumerate}
\item If $r$ is a unary rule different from $\pnbang{}_L,\vnbang{}_L$ with premise $S'$, then the game proceeds in the game state $(G\cup\{S'\},b)$.
\item {\bf Budget decrease:} If $r\in\{\pnbang{}_L,\vnbang{}_L\}$ with premise $S'$ and principal formula $\pnbang{a}A$ or $\vnbang{a}A$, then the game proceeds in the game state $(G\cup\{S'\},b-a)$.
\item {\bf Parallelism:} If $r$ is a binary rule with premises $S_1,S_2$ pertaining to a \emph{multiplicative} connective, then the game proceeds as $(G\cup\{S_1\}\cup\{S_2\},b)$.
\item {\bf \II-choice:} If $r$ is a binary rule with premises $S_1,S_2$ pertaining to an \emph{additive} connective, then \II chooses $S'\in\{S_1,S_2\}$ and the game proceeds in the game state $(G\cup\{S'\},b)$.
\end{enumerate}

\noindent A {\bf winning state} (for \I) is a game state $(H,b)$ such that all $S\in H$ are initial sequents of $\aIMALLR$ and $b\geq 0$.
\end{definition}
Plays and strategies are defined as in $\GAIMALL$. We write $\wins{}{\GAIMALLR}(H,b)$ if \I has a w.s. in the $\GAIMALLR$-game starting on $(H,b)$. The intuitive reading of $\wins{}{\GAIMALLR}(H,b)$ is: \textit{The budget~$b$ suffices to win the game $H$.} From now on, we will just say ``sequent'' and ``formula'' instead of ``extended sequent'' and ``extended formula''.

\begin{example}
Consider the state $
(\{\pnbang{1}p,\vnbang{3}q\lra p\tensor q\},5) 
$. 
In a first move, $\I$ picks $p\tensor q$ and she  finds a partition of the premises not prefixed with~$\pnbang{}$
and decides that  $\vnbang{3}q$ goes to the right premise of $\tensor_R$. So by {\bf parallelism}, the new state is
$
(\{(\pnbang{1}p\lra p),(\pnbang{1}p,\vnbang{3}q\lra q)\},5)
$. 
She now chooses to pick $\pnbang{1}p$ of the first component and, by {\bf budget decrease}, her budget decreases  and the next  state is 
$
(\{(\pnbang{1}p,p\lra p),(\pnbang{1}p,\vnbang{3}q\lra q)\},4)
$. 
Now \I picks $\vnbang{3}q$  leading to
$
(\{(\pnbang{1}p,p\lra p),(\pnbang{1}p,q\lra q)\},1)
$. 
Since both components are initial sequents and  $budget \geq 0$, this is a winning state for $\I$.
\end{example}

Similarly to  $\GAIMALL$, it is not necessary to consider all possible strategies in $\GAIMALLR$: For example, \I never needs to take the budget into account when deciding the next move. (A rule of thumb for \I could be: always play economical, i.e. avoid the rules $\pnbang{}_{L}$ and $\vnbang{}_L$ whenever possible.) It is easy to see that a~$\aIMALLR$-proof $\pr$ of a sequent $S$ translates to a w.s. in $(S,b)$ for some \textit{sufficiently large} budget $b$. Taking these observations together, one can prove the following:

\begin{theorem}[Weak adequacy for $\GAIMALLR$]
\label{theorem:wadeq}
Let $S$ be a sequent. Then

$\exists b\left(\wins{}{\GAIMALLR}(\{S\},b)\right)\quad\iff\quad\vdash_{\aIMALLR} S$
\end{theorem}
The proof is similar to the one of Theorem \ref{theorem:adeq1}. We call this theorem \emph{weak} adequacy since information about the budget $b$ is lost in the proof theoretic representation. In other words, the game $\GAIMALLR$ is more expressive than the calculus $\aIMALLR$.
To overcome this discrepancy, we now introduce a  labelled extension of $\aIMALLR$ that we call $\laIMALLR$. A $\laIMALLR$-proof is build from labelled sequents $\Gamma\lra_b A$ where $\Gamma\lra A$ is an extended sequent and $b\in\real^+$. The complete system is given in Fig. \ref{fig:lll}. 
Our aim is to prove that
$\wins{}{\GAIMALLR}(\{\Gamma\lra A\},b)\quad\iff\quad\vdash_{\laIMALLR} \Gamma\lra_b A$.

To this end, we need an analogue of Lemma \ref{lemma:independency} (independency of subgames in $\GAIMALL$) for $\GAIMALLR$. Note that crucially, the naive analogue

$\wins{}{\GAIMALLR}(\{S_1,\ldots,S_n\},b)\qquad\iff \qquad\text{for all $i\leq n$, }\wins{}{\GAIMALLR}(\{S_i\},b)$

\noindent does \emph{not} hold: Having a w.s. in $(\{S_1,\ldots,S_n\},b)$ is not the same as having w.s.'s in all $(\{S_i\},b)$'s, since the budget $b$ is shared between the subgames in $\GAIMALLR$. 
However, one can prove that there are strategies in $\GAIMALLR$ which are independent up to a partition of the budget. More precisely, 

\begin{lemma}[Quasi-independency of subgames in $\GAIMALLR$]
\label{lemma:independency2}

\noindent $\wins{}{\GAIMALLR}(\{S_1,\ldots,S_n\},b)$ \iff

\noindent $\exists b_1,\ldots,b_n\geq 0$  s.t. $\sum_{i=1}^n b_i\leq b$  and for all $i\leq n$,  $\wins{}{\GAIMALLR}(\{S_i\},b_i)$.  
\end{lemma}
\begin{proof} The direction from right to left is obvious. For the other direction, assume that $\I$ has a w.s. $\pi$ for $(\{S_1,\ldots,S_n\},b)$. We may assume wlog that this strategy is composed of strategies $\pi_1,\ldots,\pi_n$ for the subgames $S_1,\ldots,S_n$ which are both independent from each other and from the budget. In each subgame $S_i$, let $\tau_i$ be a strategy for $\II$ which maximizes the cost $b_i$ (the total decrease of the budget) of playing $\pi_i$ against $\tau_i$. Then $\wins{}{\GAIMALLR}(\{S_i\},b_i)$. Furthermore, from $\tau_1,\ldots,\tau_n$ player \II can compose a strategy $\tau$ such that when played against $\pi$ in the parallel game $\{S_1,\ldots,S_n\}$, the costs for $\I$ sum up to $\sum_{i=1}^n b_i$. Since $\pi$ is a w.s. for $(\{S_1,\ldots,S_n\},b)$, it must be the case that $\sum_{i=1}^n b_i\leq b$.
\end{proof}

We emphasize that the game rules of $\GAIMALLR$ do \emph{not} force \I to know a partition of the budget in order to play parallel subgames. Nevertheless, Lemma~\ref{lemma:independency2} tells us that finding such a partition is always possible {\it in principle} (for an omnipotent player \I).
Now we can  prove the desired correspondence.

\begin{figure}[t]
{
\[
\begin{array}{c}
\hline\mbox{labelled sequent system for }\laIMALLR\\\hline\\
 \infer[\tensor_L]{\Gamma, A \tensor B \lra_b C}
{\Gamma, A, B \lra_b C} 
\quad 
\infer[\tensor_R]{\pnbang{}\Gamma,\Delta_1, \Delta_2 \lra_{a+b} A\tensor B}
{\pnbang{}\Gamma,\Delta_1 \lra_a A & \pnbang{}\Gamma,\Delta_2 \lra_b B}
\\\\
\infer[\lolli_L]{\pnbang{}\Gamma,\Delta_1, \Delta_2, A \lolli B \lra_{a+b} C}
{\pnbang{}\Gamma,\Delta_1 \lra_a A & \pnbang{}\Gamma,\Delta_2, B \lra_b C}
\quad 
\infer[\lolli_R]{\Gamma \lra_b A \lolli B}{\Gamma, A \lra_b B}
\\\\
 \infer[\with_{L_i}]{\Gamma, A_1 \with A_2 \lra_b B}
{\Gamma, A_i\lra_b B} 
\quad 
\infer[\with_R]{\Gamma \lra_{\max\{a,b\}} A \with B}
{\Gamma \lra_a A & \Gamma \lra_b B}
\\\\
\infer[\oplus_L]{\Gamma, A \oplus B \lra_{\max\{a,b\}} C}
{\Gamma, A \lra_a C & \Gamma, B \lra_b C}
\quad 
\infer[\oplus_{R_i}]{\Gamma \lra_b A_1 \oplus A_2}{\Gamma \lra_b A_i}
\\\\
\infer[\pnbang{}_L]{\Gamma,\pnbang{a}A\lra_{c+a} C}{\Gamma,\pnbang{a}A,A\lra_{c} C}
\qquad
\infer[\vnbang{}_L]{\Gamma,\vnbang{a}A\lra_{c+a} C}{\Gamma,A\lra_c C}
\\\\
  \infer[I]{\Gamma,p \lra_0 p}{} 
 \qquad
\infer[\one_R]{\Gamma \lra_0 \one}{}
\qquad 
\infer[\zero_L]{ \Gamma,\zero\lra_0 A}{}
\qquad 
\infer[w_{\ell} (b\geq a)]{ \Gamma\lra_b A}{\Gamma\lra_a A}
\end{array}
\]}\caption{The labelled sequent system $\laIMALLR$}
\label{fig:lll}
\end{figure}

\begin{theorem}[strong adequacy for $\GAIMALLR$]
\label{theorem:adeq2}

$\wins{}{\GAIMALLR}(\{\Gamma\lra A\},b)\quad\iff\quad\vdash_{\laIMALLR} \Gamma\lra_b A$. 
\end{theorem}
\begin{proof}
($\Leftarrow$) By induction on the length of a proof 
$\pr$ of $\Gamma\lra_{b} A$. We highlight two cases.
Consider the  following two possible 
 ends for $\pr$: 
\[
\deduce[(1)\ \ ]{}{} \infer[\with_R]{\Gamma\lra_{\max\{c,d\}} C\with D}{\Gamma\lra_c C & \Gamma\lra_d D}
\qquad
\deduce[(2)\ \ ]{}{}
\infer[\tensor_R]{\Delta_1,\Delta_2\lra_{c+d} C\otimes D}{\Delta_1\lra_c C & \Delta_2\lra_d D}
\]
In both cases, by induction, there are w.s.'s $\pi_1$ and $\pi_2$ for: (1) the game states $(\{\Gamma\lra C\},c)$ and $(\{\Gamma\lra D\},d)$;  
and (2) the game states 
$(\{\Delta_1\lra C\},c)$ and $(\{\Delta_2\lra D\},d)$ respectively.
The needed w.s.'s $\pi_{\with}$ for the game state $(\{\Gamma\lra C\with D\},\max\{c,d\})$ and $\pi_{\otimes}$ for the game state $(\{\Delta_1,\Delta_2\lra C\otimes D\},c+d)$ are: 

\begin{quote} (1)$\pi_{\with}$: Choose the instance of $\with_R$ as above. By 
\II-choice, the successor game state is either $(\{\Gamma\lra C\},\max\{c,d\})$ or $(\{\Gamma\lra D\},\max\{c,d\})$. In any case, the budget in the successor state is greater or equal than both $c$ and $d$, so \I can continue playing according to $\pi_1$ resp.  $\pi_2$.\end{quote}
\begin{quote} (2)$\pi_{\otimes}$: Choose the instance of $\tensor_R$ as above. By {\bf parallelism}, the successor  state is $(\{\Delta_1\lra C,\Delta_2\lra D\},c+d)$. Use $\pi_1$ to play the subgame $\Delta_1\lra C$ and $\pi_2$ to play in $\Delta_2\lra D$. By assumption on $\pi_1$ and $\pi_2$, the total costs when playing both strategies in parallel cannot exceed $c+d$.\end{quote}
($\Rightarrow)$ By induction on the length of a strategy $\pi$
. We present only the case where Lemma \ref{lemma:independency2} is used. 
Assume that the state is $(\{\Delta_1,\Delta_2\lra C\tensor D\},b)$  and $\pi$ tells  \I to choose the instance of $\tensor_R$ with premises $\Delta_1\lra C$ and $\Delta_2\lra D$. 
 By {\bf parallelism}, the successor state is $(\{\Delta_1\lra C,\Delta_2\lra D\},b)$. Since $\pi$ is a w.s., it must contain a substrategy $\pi'$ for this state. By Lemma \ref{lemma:independency2}, we may assume that $\pi'$ is composed of substrategies $\pi_1,\pi_2$ for the game states $(\{\Delta_1\lra C\},c)$ and $(\{\Delta_2\lra D\},d)$ where $c+d\leq b$. By induction, there are $\aIMALL$-proofs 
$\pr_1,\pr_2$  for the sequents $\Delta_1\lra_c C$ and $\Delta_2\lra_d D$.  
Applying $\tensor_R$ and $w_\ell$ below $\pr_1$ and $\pr_2$, we obtain a $\aIMALL$-proof $\pr$ of $\Delta_1,\Delta_2\lra_b C\tensor D$.
\end{proof}

Let $S_b$ denote the labelled sequent corresponding to the sequent $S$ with label $b$. Given
$\Pi$ a $\laIMALLR$-proof of $S_b$, we define the many-to-one onto {\em skeleton} function $\mathcal{SK}(\Pi)$ as the $\aIMALLR$-proof $\pr$ of $S$ obtained by removing all labels and applications of $w_\ell$ from $\Pi$.  Conversely, 
we define the one-to-one {\em decoration} function $\mathcal{D}(\pr)$ as the $\laIMALLR$-proof $\Pi^\ell$ of $S_a$, obtained by assigning the label $0$ to all initial sequents 
of $\pr$ and propagating the labels downwards according to the rules of $\laIMALLR$. 
We define $\costs{\pr}:=a$. Let $\Lambda\in \mathcal{SK}^{-1}(\pr)$ be a proof of $S_c$. It is easy to see that $a\leq c$, that is, $\costs{\pr}$ is the minimal label which can be attached to $S$ w.r.t. $\pr$.
In game theoretic terms, this means the following.
\begin{theorem}
Given a $\aIMALLR$-proof $\pr$ of a sequent $S$, $\costs{\pr}$ is the smallest budget which suffices to win the game $\GAIMALLR$ on $S$ when following the strategy corresponding to $\pr$.
\end{theorem}

\begin{example}\label{ex:riddle}
Consider the following well-known riddle:
\begin{quote}
You have white and black socks in a drawer in a completely dark room. How many socks do you have to take out blindly to be sure of having a matching pair? 
\end{quote}
We can model the matching pair by the disjunction $(w\tensor w)\oplus(b\tensor b)$, and the act of drawing a random sock by the labelled formula $\pnbang{1}(w\oplus b)$. The above question then becomes:
\begin{quote}
For which $n$ is the sequent $\pnbang{1}(w\oplus b)\lra_n (w\tensor w)\oplus(b\tensor b)$   provable?
\end{quote}
The following proof shows that $n=3$ suffices: \\

\resizebox{\textwidth}{!}{
$
\infer=[3\times\pnbang{}_L]{G\lra_3 F}{
 \infer[\oplus_L]{G, w\oplus b, w\oplus b, w\oplus b\lra_0 F}{
  \infer[\oplus_L]{G, w, w\oplus b, w\oplus b\lra_0 F}{
   \infer[\oplus_R]{G, w, w, w\oplus b\lra_0 F}{
   \infer=[\otimes_R, I]{G, w, w, w\oplus b\lra_0 w\tensor w}{}
   }
   &
   \infer[\oplus_L]{G, w, b, w\oplus b\lra_0 F}{
    \infer[\oplus_R]{G, w, b, w\lra_0 F}{\infer=[\otimes_R,I]{G, w, b, w\lra_0 (w\tensor w)}{}}
    &
    \infer[\oplus_R]{G, w, b, b\lra_0 F}{\infer=[\otimes_R,I]{G, w, b, b\lra_0 b\tensor b}{}}
   }
  }
  &
  \deduce{\pr}{}
 }
}
$
}\\

where derivation $\pr$ is symmetric, $F=(w\tensor w)\oplus(b\tensor b)$ and $G=\pnbang{1}(w\oplus b)$. 
\end{example}

\subsection{The spectrum of a provable sequent}
\label{sec:spectrum}

Due to weakening on labels, many proofs in $\laIMALLR$ of labelled sequents of the form $S_b$ correspond to one skeleton proof in $\aIMALLR$ of the sequent $S$. On the other hand, $S$ may have, itself, many proofs in $\aIMALLR$, each of them having a cost, uniquely determined by the decoration $\mathcal{D}$. In this section we will consider the {\em spectrum} of such costs and prove the existence of a minimal one.
\begin{definition}
$\spec{S}:=\{\costs{\pr}\mid\text{$\pr$ is a $\aIMALLR$-proof of $S$}\}$.
\end{definition}
\noindent For example, $\spec{\pnbang{1}p,\vnbang{0.8}p,\vnbang{0.8}p\lra p\otimes p}$ consists of the numbers $\{1.6,1.8,2.6\}$ and all combinations $n+k\cdot 0.8$ where $n,k$ are natural numbers and $n\geq 2,k\leq 2$.

A subset $X\subseteq\real$ is called 
\emph{discrete} if, for every $x\in X$, there is an open interval $I\subseteq\real$ such that $I\cap X=\{x\}$.
We can prove:
\begin{theorem}
\label{theorem:propSpec}
For any sequent $S$, $\spec{S}\subseteq\real^+$ is discrete and closed.
\end{theorem}
\begin{proof}
Let $a_1,\ldots,a_n$ denote all real numbers appearing as $\pnbang{a}$ or $\vnbang{a}$ in $S$, and let us denote by $\Omega(a_1,\ldots,a_n)$ the set of all linear combinations of $a_1,\ldots,a_n$ over $\mathbb{N}$, i.e., 
$
\Omega(a_1,\ldots,a_n):=\{k_1\cdot a_1+\ldots+k_n\cdot a_n\mid k_1,\ldots,k_n\in\mathbb{N}\}
$. 
By inspecting the rules of $\laIMALLR$ and since $w_\ell$ is not applied in $\mathcal{D}(\pr)$, it is easy to see that~$\costs{\pr}\in\Omega(a_1,\ldots,a_n)$,
and hence
$\spec{S}\subseteq\Omega(a_1,\ldots,a_n)$. 
It suffices to show that each bounded monotone sequence in $\Omega(a_1,\ldots,a_n)$ is eventually constant. We may assume wlog that all the $a_i$'s are nonzero. Now consider a sequence
 $ 
(k^i_1\cdot a_1+\ldots+k^i_n\cdot a_n)_{i\geq 1}
$ 
in $\Omega(a_1,\ldots,a_n)$, and assume that $B$ is an upper bound for it (a trivial lower bound is always $0$). Pick a number $K$ such that $K\cdot\min\{a_1,\ldots,a_n\}>B$. It follows that for all $i,j$ we have $k^i_j< K$. In particular, there are only finitely many different terms in the sequence, from which our claim follows.
\end{proof}
Since any bounded below, closed set in $\real$ has an minimum, we obtain:
\begin{corollary}\label{cor:spectrum}
If $\vdash_{\aIMALLR}\Gamma\lra A$, then $\spec{\Gamma\lra A}$ has a least element. In other words, there is a smallest $b$ such that $\vdash_{\laIMALLR}\Gamma\lra_b A$.
\end{corollary}
Corollary \ref{cor:spectrum} tells us that cost-optimal strategies for all provable sequents exist, but note that the proof is not constructive.
Nevertheless, we may now define:

$
\costs{S}:=\begin{cases}\min(\spec{S})&\mbox{if }\vdash_{\aIMALLR} S\\ \infty&\mbox{otherwise}\end{cases}
$ 
\subsection{Cut admissibility}
\label{sec:cut}

So far, the  results about our game semantics $\GAIMALLR$ did not depend essentially on the chosen calculus $\aIMALLR$. We now want to relate 
proof-theoretic properties of $\aIMALLR$ and $\laIMALLR$ to the game semantics. Recall that~$\aIMALLR$ can be seen as a fragment of~$\aSELLR$, arising from the syntactic restriction that the modal operators~$\vnbang{a},\pnbang{a}$ occur only negatively in sequents (Remark \ref{remark:fragment}); consequently, there is no corresponding right rule (\emph{promotion}) in~$\aIMALLR$. This has the effect that---even though (implicit) contraction on formulas $\pnbang{a}A$ is present in~$\aIMALLR$---the proof theory of~$\aIMALLR$ is closer to $\mathbf{aIMALL}$ than to~$\aSELLR$.

$\aIMALLR$ inherits the admissibility of the following cut rule from~$\aSELLR$
\[
\infer[cut]{\pnbang{}\Gamma,\Delta_1,\Delta_2\lra C}
	{ \pnbang{}\Gamma,\Delta_1\lra A &
	\pnbang{}\Gamma,\Delta_2,A\lra C
	}
\] 
Note that, appearing both in a positive and a negative context, the cut formula~$A$ cannot contain any modal operator. 

Now, let us extend cut admissibility to the labelled system $\laIMALLR$.
Assume that both
$\pnbang{}\Gamma,\Delta_1\lra_a A$ and $\pnbang{}\Gamma,\Delta_2,A\lra_b C$ are provable in $\laIMALLR$. Forgetting  labels $a$ and $b$, we can conclude, from cut-admissibility in $\aIMALLR$, that $\vdash_{\aIMALLR}\pnbang{}\Gamma,\Delta_1,\Delta_2\lra C$. But then, $\pnbang{}\Gamma,\Delta_1,\Delta_2\lra_c C$ is also provable in $\laIMALLR$ with, \eg, 
$c=\costs{\pnbang{}\Gamma,\Delta_1,\Delta_2\lra C}$ (see Cor.~\ref{cor:spectrum}).
Hence, stating cut-admissibility in $\laIMALLR$ strongly depends on the possibility of defining a  computable function $f$ relating $c$ with the labels of the premises of the cut rule. 
We show that $f(a,b)=a+b$ is the {\em minimal} such function.
\begin{theorem}\label{thm:cutAdm}
For $f(a,b)=a+b$, the following cut rule is admissible in $\laIMALLR$:
\[
\infer[cut_\ell]{\pnbang{}\Gamma,\Delta_1,\Delta_2\lra_{f(a,b)} C}
	{ \pnbang{}\Gamma,\Delta_1\lra_a A &
	\pnbang{}\Gamma,\Delta_2,A\lra_b C
	}
\]
Moreover, whenever $cut_\ell$ is admissible w.r.t. a given $f'$, then $a+b\leq f'(a,b)$.
\end{theorem}
\begin{proof}
For cut admissibility, one can follow the standard cut reduction strategy of $\mathbf{aIMALL}$ and observe that it is compatible with the proposed labelling of the cut rule. Consider for instance  the following reduction (note that $\max\{a+c,a+d\}=a+\max\{c,d\}$):\\

\resizebox{.7\textwidth}{!}{$
\infer[cut_\ell]{\pnbang{}\Gamma,\Delta_1,\Delta_2\lra_{a+\max\{c,d\}} C\with D}
	{ \pnbang{}\Gamma,\Delta_1\lra_a A &
	\infer[\with_R]{\pnbang{}\Gamma,\Delta_2,A\lra_{\max\{c,d\}} C\with D}
		{ \pnbang{}\Gamma,\Delta_2,A\lra_c C &
		\pnbang{}\Gamma,\Delta_2,A\lra_d D
		}
	}
$}
\\$\rightsquigarrow$
\\

\noindent\resizebox{.95\textwidth}{!}{
$
\infer[\with_R]{\pnbang{}\Gamma,\Delta_1,\Delta_2\lra_{\max\{a+c,a+d\}} C\with D}
	{
	\infer[cut_\ell]{\pnbang{}\Gamma,\Delta_1,\Delta_2\lra_{a+c} C}
		{  \pnbang{}\Gamma,\Delta_1,\lra_a A &
		\pnbang{}\Gamma,\Delta_2,A\lra_c C
		} &
	\infer[cut_\ell]{\pnbang{}\Gamma,\Delta_1,\Delta_2\lra_{a+d} D}
		{  \pnbang{}\Gamma,\Delta_1,\lra_a A &
		\pnbang{}\Gamma,\Delta_2,A\lra_d D
		}
	}		
$}
\\

For the minimality, let $p,q$ be distinct propositional variables. For any $a,b\in\real^+$ we have proofs of $\pnbang{a}p\lraS{a}p\quad\text{and}\quad p,\pnbang{b}q\lraS{b}p\tensor q$.  Applying cut, we get $\pnbang{a}p,\pnbang{b}q\lraS{c} p\tensor q$. Now, $\pnbang{a}p,\pnbang{b}q\lraS{c} p\tensor q$ is provable (without cut) only if $a+b\leq c$. Hence if $f$ makes the cut rule admissible, $a+b\leq f(a,b)$.
\end{proof}

One can easily show that also  weakening in the antecedent is admissible in~$\laIMALLR$ and does not lead to an increased label. Similarly, generalized axioms $\Gamma,A\lra_0 A$ are admissible: Appearing both positively and negatively, $A$ does not contain modal operators, and hence $\costs{\Gamma,A\lra A}=\costs{A\lra A}=0$.
\begin{example} \label{ex:tsystem}
Consider a labelled transition system $(T,\rede{})$
where $T$ is a set of states and $\rede{} \subseteq T \times \realP \times T$ is the transition relation on states. In $(t_i,a_i,t_i')\in \rede{}$, simply written as $t_i \rede{a_i}t_i'$,   $a_i$ is interpreted as the time needed for the transition to happen. 
We  use distinct propositional variables to 
represent states. Moreover, 
the  formula $\pnbang{a_i}(t_i\limp t_i')$  models the transition $t_i \rede{a_i}t_i'$. We shall use $\Delta_{\rede{}}$ to denote the set of such formulas. 
Given two sets of states $S_{start},S_{end}\subseteq T$, it is easy to see that the following sentences are equivalent:
\begin{enumerate}
\item From every state in $S_{start}$, there is a state in $S_{end}$ reachable in time $\leq a$
\item $\wins{}{\GAIMALLR} (\{\Delta_{\rede{}}, \bigoplus S_{start}\lra \bigoplus S_{end}\},a)$
\end{enumerate}
Hence by Theorem \ref{theorem:adeq2}, both are equivalent to
\begin{enumerate}\setcounter{enumi}{2}
\item $\vdash_{\laIMALLR} \Delta_{\rede{}}, \bigoplus S_{start}\lra_a \bigoplus S_{end}$.
\end{enumerate}

One common way to obtain (1) is by finding a set of intermediary states $S_i$ and a splitting of the time $a_1+a_2= a$ such that we can go from each state in $S_{start}$ to some state in $S_i$ in time $a_1$, and from each state in $S_i$ to some state in $S_{end}$ in time $a_2$. In terms of (3), this strategy corresponds to a cut: Assume we have proofs~$\Xi_1$ and $\Xi_2$ of the sequents~$\Delta_{\rede{}}, \bigoplus S_{start}\lra_{a_1} \bigoplus S_i$ and~$\Delta_{\rede{}}, \bigoplus S_i\lra_{a_2} \bigoplus S_{end}$. By cut admissibility (Theorem \ref{thm:cutAdm}) we obtain the desired~$\Delta_{\rede{}}, \bigoplus S_{start}\lra_{a_1+a_2} \bigoplus S_{end}$ as the result of the ``concatenation''of the paths encoded in $\Xi_1$ with the paths encoded in $\Xi_2$.
\end{example}

\section{Alternative cost structures}
\label{sec:general}

We have used non-negative real numbers  for representing costs and budgets, together with   basic operations  for {\em accumulating} ($+$)  and {\em comparing} ($\geq$) them. 
This allowed us  to give 
a more interesting perspective of \emph{resource  consumption} in linear logic: we know that the cost of using  a formula  marked with  cost $3$ is not the  same  as  derelicting a formula marked with  cost $7$. 
A natural question that arises is whether it is possible to consider other systems 
governing the way we understand \emph{costs} and \emph{budgets}. In this section, we consider sequent systems $\laIMALLA$ in which the real numbers of $\laIMALLR$ (see Fig. \ref{fig:lll}) are replaced by elements of a semiring $\cK$. As we shall see, the structure of $\cK$ determines the  behavior of the system and the interpretation of costs and budgets.

\newcommand{\eA}{1_{\mathcal{A}}}

A {\em commutative semiring} is a tuple $\cK = \langle \cA,\plusA,\timesA$ $, \botA,\topA\rangle$ satisfying:
 (S1) $\cA$ is a set and $\botA,\topA \in \cA$;
 (S2)  $\plusA$ and $\timesA$ are binary operators that make the triples $\langle A, \plusA , \bot_\cA\rangle$ and 
  $\langle A, \timesA , \top_\cA\rangle$ commutative monoids;
  (S3)   $\times_{\cA}$  distributes over $\plusA$ (i.e., $a \timesA (b \plusA c ) = (a\timesA b) \plusA (a \timesA c)$);
  and (S4) $\bot_\cA$ is  absorbing  for $\timesA$ (i.e., $a  \timesA \bot_\cA = \bot_\cA$);  
 $\cK$ is {\em absorptive} if it additionally satisfies (S5)
 $a \plusA (a \timesA b) = a$; in absorptive semirings, $\plusA$ is idempotent, that is, $a\plusA a=a$. This allows for the definition of the following {\em partial order}:  $a \leqA b$ 
  iff $a\plusA b = b$
  (and then,  $a \timesA b \leqA a$);
    an absorptive semiring $\cK$ is {\em idempotent} whenever its $\timesA$ operator is idempotent.

Absorptive semirings satisfy some additional properties \cite{DBLP:conf/ecai/BistarelliG06}: 
  $\botA$ (resp. $\topA$) is the bottom (resp. top) of $\cA$; 
 $a \plusA \topA = \topA$; 
 $\plusA$ coincides with the $\lubA$ (least upper bound) operator;
 if $a\plusA b \in \{a,b\},\,\forall a,b\in\cA$ then ($\cA, \leqA$) is a total order; 
  $a \timesA b \leqA \glbA(a,b)$, where $\glbA$ is the greatest lower bound operator;
  if   $\cK$ is idempotent, then $\plusA$ distributes over $\timesA$ and  $\timesA$ coincides with $\glbA$.

We  identify  costs as elements of $\cA$.
We can naturally 
consider $\topA$ (resp. $\botA$) 
as
the ``best'' (resp. ``worst'') cost.
Dually,   $\topA$ (resp. $\botA$) is the ``worst'' (resp. ``best'') budget. 
Also, we expect the {\em accumulating} operator to be commutative and associative (S2). 
Moreover, 
{\em accumulating} costs gives rise to a ``worse'' cost (S5). Hence, the $\times_{\cA}$ operator is used to combine costs ($+$, on $\real^+$,  in Fig.  \ref{fig:lll}).
On the other hand,    $+_{\cA}$ is used to select which is the ``best'' value,   in the sense that $a\plusA b = a$ iff $b \leqA a$ iff $a$ is ``better'' than $b$ (i.e.,  $\preceq_\cA$ will replace  $\geq$ in Fig.  \ref{fig:lll}). 
Finally, we generalize   $\max$ (in Fig.  \ref{fig:lll}) 
as $\glbA$.
 As mentioned above,  in the case of idempotent semirings, $\timesA$ coincides with the $\glbA$ while in the non-idempotent  case accumulating costs often gives a ``worse'' result than the $\glbA$.

Note that the rules $\vnbang{a}_L$ and  $\pnbang{a}_L$,  in Fig. \ref{fig:lll},
the budget $c$ in the conclusion must be of the form $a+b$. In the particular case of $\real^+$, 
we know that 
$b = c -a$ whenever $c \geq a$. Hence,  from a conclusion with   budget $c$  we obtain a premise with {\bf decreased} budget $c-a$.  
In the general case,
we guarantee that such splitting of the budget (also present in rules  $\otimes_R$ and $\limp_L$)  is possible  
by requiring $\cK$ to be invertible in the following sense:
$\cK$ is 
{\em invertible} if for all $b\leqA a$, the set  $\cI(b,a)= \{x \in \cA \mid a \timesA x =b\}$ is non-empty and admits a minimum. We then denote this minimum by $b \divA a$.
Observe that,  in all our examples, if $b \leqA a $ then the set $\cI(b,a)$   is a singleton except when $a=b=\botA$. In that case,   $\cI(b,a) = \cA$ and   we set $\botA \divA \botA = \botA$.
In Remark \ref{rem:bot} we explain and clarify this choice.

In what follows,  $\cK$ will always denote an absorptive and invertible 
semiring.

\begin{definition}[System $\laIMALLA$]\label{def:new-system}
Let $\cK=\langle \cA,\plusA,\timesA, \botA,\topA\rangle$ be an absorptive and invertible semiring. 
The system $\laIMALLA$ is obtained from $\laIMALLR$ (Fig. \ref{fig:lll}) by replacing $0$ with $\topA$, $+$ with $\times_\cA$, $\max$ with $\glb_{\cA}$, and $\geq$ with $\preceq_{\cA}$. Similarly, we obtain   $\aIMALLK$
as a generalization of $\aIMALLR$ (Fig. \ref{fig:ll}). 
\end{definition}
Just as $\aIMALLR$ can be seen as a fragment of $\aSELLR$, the system $\aIMALLK$ is a fragment of $\aSELLA$, i.e. affine subexponential linear logic with subexponentials taken from the set $\mathcal{K}\times\{\mathfrak{u},\mathfrak{b}\}$. We omit the (rather straightforward) formulation of the corresponding game semantics.

Next we present some instances of $\laIMALLA$
and their intended behavior.

\begin{example}[Costs]\label{ex:prices}
 $\cK_c = \langle  \cRpi, \min_\cR, +_\cR, \infty, 0  \rangle$, where $\cRpi$ is the completion of $\cRp$ with $\infty$, reflects  the meaning of costs   in Section \ref{sec:costs}. If $a,b\not=\infty$ and $b \geq a$ (i.e., $b \preceq_\cA a$), 
 there is a  unique way of splitting $b$
 into $a + b'$, namely, $b' = b -a$ (i.e., $b' = b \div_\cA a$).
Alternatively,  we may interpret the elements in $\cK_c$ as 2D areas. Then, a label $b\neq \infty$ in a sequent  can be understood as the total area available to place some objects. Each time an object of size $a$ is placed (using $\pnbang{a}{}_L$ or $\vnbang{a}{}_L$)  we observe, bottom-up,   that  the  total area is decreased to $b- a$. 
\end{example}
  \begin{remark}[Meaning of $\div$ and $\infty$]\label{rem:bot} 
  Consider the semiring   $\cK_c$ above (where $\botA = \infty$ and $\topA = 0$). 
If the  label in the sequent is $b=\botA$, regardless the value $a$ in an application of $\vnbang{a}_L$ or  $\pnbang{a}_L$,
the premise will be labelled with $\infty$. This is because, according to our definition, $\botA \div \botA = \botA$. This makes sense since we select the most ``generous'' budget to continue the derivation. Of course, smaller suitable budgets  are also allowed due to  rule $w_{\ell}$. 
For instance, the sequent $\vnbang{\bot}p, \vnbang{\bot}(p \limp q) \lra_b q$
is provable in $\laIMALLAC$  only if $b=\botA$. The same sequent (removing the label $b$) is also provable in \aSELLA. 
Note that if we decree that  $\botA \div \botA = \topA$ (as in  \cite{DBLP:conf/ecai/BistarelliG06}),
then the sequent above would not be provable for any $b$.  
\end{remark}

\begin{example}[Protected resources]
Let $\cK_{c/p}=\langle \{\public,\confidential\}, +, \times, \public, \confidential \rangle$
and define $a+b = \public$ iff $a = b = \public$ and $a \times b = \confidential$ iff 
$a=b=\confidential$. 
The intuition is that 
 $\pnbang{\public}F$ represents {\em public}  information  (and then not confidential) and  
$\pnbang{\confidential}F$ represents {\em secret} information. 
Observe that no derivation of $\Gamma, \pnbang{\public}F \lra_{\confidential} G$ can apply $\pnbang{}_L$ on $\pnbang{\public}F$ (since $ \confidential\not\preceq \public$). 
 This means that only confidential (or protected) resources can be used in such a derivation. 
	Alternatively, we can show that, if  $\Gamma \lraS{\confidential} G$ is provable then 
  $\Gamma' \lraS{\confidential} G'$ is also provable where $\Gamma'$ is as $\Gamma$ but replacing any formula of the form 
	$\pnbang{\public}F$ with the constant $\one$ (similarly for $G$ and $\vnbang{\public}F$). 
$\cK_{c/p}$  is nothing less that the structure
$S_c = \langle \{\false,\true\}, \vee,\wedge,\false,\true \rangle$ \cite{DBLP:journals/constraints/BistarelliMRSVF99}.
\end{example}

\begin{example}[Maximum amount of resources] Consider 
now the situation where labels in sequents  represents a certain amount of computational resources, e.g., RAM, 
	available to process a series of tasks. Moreover, 
	let us interpret $\vnbang{c}F$ as the fact that, 
	  in order to  produce  $F$, $c$ resources need to be used.
	As expected, once $F$ is produced, the $c$ resources can be released and  freed to be used in other tasks. The idea is to know what is the least amount of resources $b$ s.t. some  jobs $\Gamma$   can be all of them executed, sequentially if needed. 
	  
     Consider    $\cK_{\max} = \langle  \cRpi, \min, \max, \infty, 0  \rangle$ 
where  $b \div a = b$ (if $b \geq a$). Let $t_1, t_2$ be atomic propositions representing tasks and let $\Gamma = \{\vnbang{a}t_1, \vnbang{c}t_2\}$. Clearly, 
the sequents $\Gamma \lra_b t_1 \otimes t_2$
and $\Gamma \lra_b t_1 \with t_2$
are both provable if 
$b =\max(a,c)$. Of course, if we start with more resources, $e.g., b=a+c$, the sequent is still provable (rule $w_{\ell}$). 
     Interestingly, from the point of view of costs, the difference between concurrent  ($\otimes$) and 
     sequential choices ($\with$) vanishes in this particular scenario, since $\cK_{\max}$ is idempotent (and hence $\glbA$ and $\timesA$ coincide). 
\end{example}

\begin{example}[Transition systems revisited]
Consider the 
formulas of the shape $\pnbang{a_i}(t_i\limp t_i')$
and the sequent $\Delta_{\rede{}}, t\lra_b t'$
in  Example \ref{ex:tsystem}. 
The interpretation there, of   $b$  as the time needed to observe a transition from $t$ to $t'$,  can be captured with the semiring 
$\cK_c$ (Example \ref{ex:prices}). As expected, according to $\plusA$ (and then $\leqA$), we prefer ``faster'' paths when there are different ways of going from $t$ to $t'$. 

 Another possible interpretation for $b$  is the {\em probability} of the different \emph{independent} events (transitions) to happen. Hence, given a specific path from $t$ to $t'$, the possible values for   $b$ must be less or equal to the product of the probabilities $a_i$ involved in that path. 
 This behavior can be captured with the probabilistic semiring \cite{DBLP:journals/constraints/BistarelliMRSVF99} 
$\cK_p = \langle  [0,1], \max, \times, 0 ,1  \rangle$. 

For yet another example of $\cK_p$, consider the typical probabilistic choice in process calculi: the process $P +_{\alpha} Q$   chooses $P$ with probability $\alpha$ and $Q$ with probability $1-\alpha$. Following the process-as-formulas interpretation 
\cite{DBLP:conf/elp/Miller92,DBLP:journals/mscs/DengSC16}, relating  process constructors with logical connectives and reductions with proof steps, the system $\laIMALLAP$  offers a very natural interpretation of the process $P +_{\alpha} Q$ as the formula
	$(\vnbang{\alpha}P) \with (\vnbang{1- \alpha}Q)$, that we can write as $P \with_\alpha Q$. For instance, if $\Gamma = \{t_1 \with_\alpha t_2, t_1 \limp t_3, t_2\limp t_4\} $, 
	then, 
	the sequent  $\Gamma\lra_b t_3$ 
	(resp.  $\Gamma\lra_{b} t_4$) is  provable
	whenever $\alpha \geq b$ (resp. $ 1- \alpha \geq b $). 
\end{example}

\section{Modalities in positive contexts}
\label{section:pos}

We have considered modalities appearing only in negative polarity. 
In this section, we show some problems and limitations
that arise when trying to extend the labelled sequent approach to consider also positive occurrences of modalities 
as in the full system of subexponential linear logic (see e.g., \cite{nigam09ppdp}).
Let us call 
\laSELLR  the system resulting from \laIMALLR~ by 
adding the following \emph{labelled promotion rules}
\[
\infer{\Gamma\lra_b \vnbang{a} A}{\Gamma^{\leqvn{a}}\lra_b A}\qquad
\infer{\Gamma\lra_b \pnbang{a} A}{\Gamma^{\leqpn{a}}\lra_b A}
\]
where $\Gamma^{\leqvn{a}}$ denotes all formulas in $\Gamma$ which are of the form $\vnbang{c} B$ or $\pnbang{c} B$  and $a \geq c$; and 
$\Gamma^{\leqpn{a}}$ 
denotes all formulas in $\Gamma$ which are of the form $\pnbang{c} B$ where $a \geq c$. These rules follow the standard formulation of the promotion rule in  $\SELL$: the promotion of $!^{a}A$ requires all formulas of the context to be of the form $!^{c}B$
where $a \preceq c$  and $\preceq$ is the underlying  preorder on the subexponential signature.

The following result shows that it is not possible to define a labelled cut rule for \laSELLR
where the label of the conclusion depends exclusively on the labels of the premises.
\begin{theorem}\label{thm:impossible} There is no function $f:\real^+\times \real^+\rightarrow\real^+$ such that the rule
\[
\infer[cut]{\pnbang{}\Gamma,\Delta_1,\Delta_2\lra_{f(a,b)} C}
	{ \pnbang{}\Gamma,\Delta_1\lra_a A &
	\pnbang{}\Gamma,\Delta_2,A\lra_b C
	}
\]
is admissible in $\laSELLR$.
\end{theorem}
\begin{proof}
Let $p,q$ be different propositional variables, and let $A^{\tensor n}$ denote the $n$-fold multiplicative conjunction of a formula $A$. The sequents
\[\pnbang{1/k}p\lra_{a}\pnbang{1/k}p^{\otimes (k\cdot a)}\qquad\text{and}\qquad \pnbang{1/k}p^{\otimes (k\cdot a)}\lra_b p^{\otimes(k\cdot k\cdot a\cdot b)} 
\]
are provable in $\laSELLR$ for all natural numbers $a,b,k$. The smallest label~$f$ which makes their cut conclusion
$ \pnbang{1/k}p\lra_f p^{\otimes(k\cdot k\cdot a\cdot b)}
$ 
provable in~$\laSELLR$ is~$k\cdot a\cdot b$, which is not a function on the premise labels~$a,b$.
\end{proof}
Note that Thm. \ref{thm:impossible} leaves open the possibility that cut is admissible w.r.t. a function $f$ which takes more information of the premises into account than just their labels. Please refer to the appendix   for a more detailed discussion. 
\section{Concluding remarks and future work}\label{sec:conc}

We have introduced game semantics for fragments of (affine intuitionistic) linear logic with subexponentials (\SELL~\cite{danos93kgc,nigam09ppdp,DBLP:journals/tplp/PimentelON14}), culminating in labelled extensions of such systems so that
$\Gamma \lra_{b} A$ is interpreted as: ``Resource $A$ can be obtained from the resources $\Gamma$ with a budget $b$'' or, alternatively, ``The budget $b$ suffices to win the game $\Gamma \lra A$''.  For achieving that, we proposed a new interpretation for the {\em dereliction} rule, opposing to the standard controls in the {\em promotion rule}:  derelicting on $\vnbang{a}B,\pnbang{a}B$ means ``paying $a$ to obtain (a copy of) $B$''. Hence our games and systems offer a neater  control of the resources appearing {\em negatively} on sequents. 

There are several ways of extending and continuing this work. First of all, as signalized in Sec.~\ref{section:pos}, the quest of extending the cost conscious reasoning to modalities occurring {\em positively} in sequents is not trivial. Despite the obvious game interpretation of promotion that could be given in the style of~\cite{DBLP:conf/tableaux/FermullerL17}, Thm.~\ref{thm:impossible} shows that this would {\em not} be followed with a proof theoretical notion of cut-elimination, due to the impossibility of defining a functional notion of the cut-label. In the appendix we discuss some possible paths to trail in this direction. On the other side, a philosophical discussion on the need of compositionally of dialogue games driven by a cut rule can also be done~\cite{dutilh18}. 

Finally, we expect that the study of costs of proofs and cut-elimination in labelled fragments of \SELL\ may indicate a relationship between labels and bounds of computation~\cite{DBLP:journals/pacmpl/AccattoliGK18}, as well as give a different approach to study the complexity of cut-elimination process, specially in the multiplicative-(sub)exponential fragment~\cite{DBLP:journals/tcs/Strassburger03,DBLP:journals/tocl/StrassburgerG11}.

\newpage
\appendix
\section{Appendix}

\subsection{Cut elimination for $\laSELLK$}\label{app:cut-elim}
Consider the labelled sequent system $\laSELLK$ built from 
$\laIMALLA$ (see Def.~\ref{def:new-system}) by dropping the restriction on occurrences  of modalities and  adding the following 
  \emph{labelled promotion rules}
(see Section \ref{section:pos})
\[
\infer{\Gamma\lra_b \vnbang{a} A}{\Gamma^{\leqvnk{a}}\lra_b A}\qquad
\infer{\Gamma\lra_b \pnbang{a} A}{\Gamma^{\leqpnk{a}}\lra_b A}
\]
where $\Gamma^{\leqvnk{a}}$ denotes all formulas in $\Gamma$ which are of the form $\vnbang{c} B$ or $\pnbang{c} B$  and $a \preceq_\cA c$; and 
$\Gamma^{\leqpnk{a}}$ 
denotes all formulas in $\Gamma$ which are of the form $\pnbang{c} B$ where $a \preceq_\cA c$.
 We shall explore 
different fragments and (admissible) cut-like rules that can be proposed for such a calculus. For concreteness, we 
consider the case $\cK = \cK_c$ (Example \ref{ex:prices}). But note, however, the discussion below applies for any absorptive, invertible semiring $\cK$. 

We start by observing that the inclusion of ``worse costs''  ($\infty$ in the reals, $\botA$ in the semiring) entails a trivial labelling 
that makes cut admissible. In the following theorem,  the cut formula $F$ is an arbitrary formula (containing, possibly, positive and/or negative occurrences of 
the modalities $\pnbang{a}$ or $\vnbang{a}$). 

\begin{theorem}[$cut_\infty$ Rule]
The following rule is admissible in $\laSELLK$
\[
\infer[cut_\infty]{\pnbang{}\Gamma,\Delta_1,\Delta_2 \lraS{\infty} C}{
 \deduce{\pnbang{}\Gamma,\Delta_1 \lraS{a} F}{}&
  \deduce{\pnbang{}\Gamma,\Delta_2, F \lraS{b} C}{}
 }
\]
\end{theorem}
The proof follows the same steps of the cut-elimination proof for \SELL, 
using natural extensions of invertibility and permutability of rules to the labelled case.

It is worth noticing that the sole responsible for the impossibility result of Thm.~\ref{thm:impossible} is the explosive combination of the use of tensor/implication and contraction, that is, \SELL's multiplicative-(sub)exponential fragment. Hence, limiting the occurrence of one or the other leads to more amenable results.
For example, Thm. \ref{thm:cutAdm} can be straightforwardly 
extended for 
formulas not containing the modality $\pnbang{a}$ (but $\vnbang{a}$ may occur). 

\begin{theorem}[Linear cuts]
Let $F$ be a formula with no occurrences of 
$\pnbang{a}$. Then, the following rule is admissible in $\laSELLK$
\[
\infer[cutL]{\pnbang{}\Gamma,\Delta_1,\Delta_2 \lraS{a + b} C}{
 \deduce{\pnbang{}\Gamma,\Delta_1 \lraS{a}F}{}&
 \deduce{\pnbang{}\Gamma,\Delta_2, F \lraS{b}C}{}&
}
\]
Moreover, if $\Gamma \lraS{a} C$ is provable using cutL, then there is a cut-free proof of 
$\Gamma \lraS{a'}C$ with $a \geq a'$.
\end{theorem}
\begin{proof}
The cut-elimination procedure is rather standard. Let us present the case when the cut formula is $\vnbang{c}F$: 
\[
\infer[cutL]{\pnbang{}\Gamma,\Delta_1,\Delta_2\lraS{a+b+c} C}{
 \infer{\pnbang{}\Gamma,\Delta_1 \lraS{a} \vnbang{c}F}{(\pnbang{}\Gamma,\Delta_1)^\leqvn{c} \lraS{a} F}
 &
 \infer{\pnbang{}\Gamma,\Delta_2, \vnbang{c}F \lraS{b+c} C}{\pnbang{}\Gamma,\Delta_2, F \lraS{b} C}
}
\]
reduces to
\[
\infer[cutL]{\pnbang{}\Gamma,\Delta_1, \Delta_2\lraS{a+b} C}{
 \deduce{(\pnbang{}\Gamma,\Delta_1)^\leqvn{c}  \lraS{a} F}{}
 &
 \deduce{\pnbang{}\Gamma,\Delta_2, F \lraS{b} C}{}
}
\]
\end{proof}
Still, forcing cut formulas to be linear seems to be a very severe restriction
to impose. A better approach is given by keeping an exact track of the use of contraction in the cut-elimination process.  
The idea is that, if proving $F$ costs $a$, then any use of $F$  
must pay this ``extra cost''. In order to keep track of this extra cost, we introduce the following notation.

\begin{definition}
Let $\cE=\{a_b\mid a,b\in \cRpi\}$ be such that 
\begin{enumerate}
\item $a_b+_\cE c_d=a+b+c+d$.
\item $a_b \geq_\cE a_c$ (i.e., the ordering $\geq_\cE$ ignores the subindices).
\item $a_b >_\cE c_d$ iff $a>c$.
\end{enumerate}
For any formula $F\in\laSELLK$, we define $[F]_c$ as the formula that substitutes any 
modality $\pnbang{a_b}{}$ with $\pnbang{a_{b+c}}$.
\end{definition}
Hence $\laSELLK$ can be slightly modified so that sequent labels belong to  $\cRpi$, while modal labels belong to $\cE$. Due to the ordering above, the promotion of $\pnbang{a_0}$ 
has the same effect/constraints that the promotion of $\pnbang{a_b}$. However, the dereliction of the latter requires a greater budget ($a+b$ instead of $a$). Moreover, the equivalence $\pnbang{a_b}F \equiv \pnbang{a_c}F$ can be proven, each direction requiring a different budget.
Finally, note that $\cE_0=\{a_0\mid a\in \cRpi\}\simeq \cRpi$, that is, each element $a\in\cRpi$ can be seen as the equivalence class of $a_0$ in $\cRpi\times \cRpi$ modulo $\cRpi$.
We will abuse the notation and continue representing the resulting system by $\laSELLK$, also unchanging the representation of sequents.

The following lemma has a straightforward proof.
\begin{lemma}
If $\Gamma, [F]_c \lraS{b}G$ then 
$\Gamma, F \lraS{b'}G$
with $b \geq b'$. More generally,  if $\Gamma, [F]_c \lraS{b}C$ and $c \geq c'$ then
$\Gamma, [F]_{c'} \lraS{b'}C$ with $b \geq b'$. 
\end{lemma}
The next definition restricts the appearance of unbounded modalities 
only under linear implication.
\begin{definition}[$\limp$-linear]
We say that $F$ is  $\limp$-linear if for all subformulas of the form $A \limp B$ in $F$, $A$ does not have occurrences 
of $\pnbang{a}$. 
\end{definition}
The following result presents the admissibility of an extended form of the cut rule, where the budget information from the left premise is passed to the cut-formula in the right premise. Observe that the label
of the conclusion is now a function of the labels of the premises. Moreover, the cut-reduction is {\em label preserving}, meaning that the budget monotonically decreases in the cut-elimination process.
\begin{theorem}[$\limp$-linear cut]
The following rule is admissible
\[
\infer[cut_{LL}\mbox{\quad $F$ is a $\limp$-linear formula}]{\pnbang{}\Gamma,\Delta_1,\Delta_2 \lraS{a + b} C}{
 \deduce{\pnbang{}\Gamma,\Delta_1 \lraS{a}F}{}&
 \deduce{\pnbang{}\Gamma,\Delta_2, [F]_a \lraS{b}C}{}&
}
\]
Moreover, if $\Gamma \lraS{a} C$ is provable using $cut_{LL}$, then there is a cut-free proof of 
$\Gamma \lraS{a'}C$ with $a\geq a'$.

\end{theorem}
\begin{proof}
We will illustrate some cases. 
\begin{itemize}
 \item Note that: $[\pnbang{a_b}F]_c = \pnbang{a_{b+c}}[F]_c$;  the promotion of $\pnbang{a_b}F$, bottom-up, results in a context of 
 $\pnbang{}$ formulas (that can be contracted at will);
 and   the dereliction of $\pnbang{a_b}[F]_c$ decreases the budget in $a + b$. Hence, 
  \[
 \infer{\pnbang{}\Gamma,\Delta_1,\Delta_2 \lraS{a + b + 2c+d}{C}}{
   \infer{\pnbang{}\Gamma,\Delta_1 \lraS{c}\pnbang{a_b}F}{(\pnbang{}\Gamma)^\leqpn{a_b}\lraS{c}F}&
   \infer{\pnbang{}\Gamma, \Delta_2, \pnbang{a_{b+c}}[F]_c \lraS{a+b+c+d}C}{\pnbang{}\Gamma, \Delta_2, [F]_{c}, \pnbang{a_{b+c}}[F]_c \lraS{d}C}
 }
 \]
reduces to
\[
   \infer{\pnbang{}\Gamma,\Delta_1,\Delta_2 \lraS{2c+d}C}{
    \deduce{(\pnbang{}\Gamma)^\leqpn{a_b} \lraS{c}F}{} &
     \infer{\pnbang{}\Gamma,\Delta_2, [F]_{c}\lraS{c+d}C}{
      \deduce{\pnbang{}\Gamma \lraS{c} \pnbang{a_b}F}{}&
      \deduce{\pnbang{}\Gamma,\pnbang{a_{b+c}}[F]_c, \Delta_2, [F]_{c}\lraS{d}C}{}
     }
    }
 \]
where the  ``extra cost'' $a_b$ disappears after the reduction. 
 \item Note that $[F\otimes G]_c = [F]_{c} \otimes [G]_c$. Here, let $c = c_1 + c_2$: 
 \[
 \infer{\pnbang{}\Gamma,\Delta_1,\Delta_2 \lraS{b+c}C}{
  \infer{\pnbang{}\Gamma,\Delta_1 \lraS{c} F \otimes G}{
   \deduce{\pnbang{}\Gamma,\Delta_1' \lraS{c_1}F}{} &
   \deduce{\pnbang{}\Gamma,\Delta_1'' \lraS{c_2}G}{} &
  } &
  \infer{\pnbang{}\Gamma,\Delta_2, [F \otimes G]_{c} \lraS{b}C}{\pnbang{}\Gamma,\Delta_2, [F]_c, [G]_c \lraS{b}C}
 }
 \]
 reduces to
\[
 \infer{\pnbang{}\Gamma,\Delta_1,\Delta_2 \lraS{b  +c}C}{
  \deduce{\pnbang{}\Gamma,\Delta_1' \lraS{c_1} F}{} &
  \infer{\pnbang{}\Gamma,\Delta_1 '',\Delta_2, [F]_{c_1} \lraS{b+c_2}C}{
   \deduce{\pnbang{}\Gamma,\Delta_1'' \lraS{c_2} G}{} &
   \deduce{\pnbang{}\Gamma,\Delta_2,[F]_{c_1}, [G]_{c_2} \lraS{b}C}{}
  }
 }
 \]
 It is worth noticing that in the first derivation, the cost $c=c_1 + c_2$ is ``charged'' to  $F\otimes G$ 
(in the formula $[F \otimes G]_c$)
while in the second one, in a  finer way, the cost $c_1$ is charged to $F$ and $c_2$ to $G$. 
 \item The case of implication explains the restriction we impose. Here $b = b_1 + b_2$:
 \[
\infer{\pnbang{}\Gamma,\Delta_1,\Delta_2 \lraS{c + b}C }{
  \infer{\pnbang{}\Gamma,\Delta_1 \lraS{c} F\limp G}{
   \deduce{\pnbang{}\Gamma,\Delta_1,F  \lraS{c} G}{}&
   }&
   \infer{\pnbang{}\Gamma,\Delta_2, [F \limp G]_c \lraS{b }C }{
     \deduce{\pnbang{}\Gamma,\Delta_2' \lraS{b_1}[F]_{c}}{}&
     \deduce{\pnbang{}\Gamma,\Delta_2'', [G]_c\lraS{b_2} C}{}
   }
}
 \]
 reduces to 
 \[
 \infer{\pnbang{}\Gamma,\Delta_1,\Delta_2 \lraS{c + b} C}{
   \deduce{\pnbang{}\Gamma,\Delta_2' \lraS{b_1} F}{}&
   \infer{\pnbang{}\Gamma,\Delta_1,\Delta_2'', [F]_{b_1} \lraS{c + b_2}C}{
    \deduce{\pnbang{}\Gamma,\Delta_1, [F]_{b_1} \lraS{c} G}{} &
    \deduce{\pnbang{}\Gamma,\Delta_2'', [G]_{c} \lraS{b_2} C}{}
   }
 }
 \]
Note that the reduction above is correct since $F$ does not have 
occurrences of $\pnbang{a}$ and then  $[F]_c = [F]_{b_1}=F$. 
\end{itemize}
\end{proof}
This kind of analysis seems to be related with {\em flowgraphs} in \MELL~\cite{DBLP:journals/tcs/Strassburger03,DBLP:journals/tocl/StrassburgerG11}. 
 \label{sec:app}

\end{document}